\documentclass[reqno,12pt]{amsart}

\usepackage{graphicx}
\usepackage{subfigure}
\usepackage{epstopdf}

\usepackage{amssymb}

\usepackage{amsthm}
\usepackage{amsmath}
\usepackage{amsfonts}
\usepackage{comment}

\usepackage{mathrsfs}
\usepackage{color}
\usepackage{bm}

\newtheorem{theorem}{Theorem}[section]
\newtheorem{lemma}[theorem]{Lemma}
\newtheorem{proposition}[theorem]{Proposition}

\newcommand{\N}{\mathbb{N}}

\newcommand{\R}{\mathbb{R}}
\newcommand{\C}{\mathbb{C}}

\theoremstyle{definition}
\newtheorem{definition}[theorem]{Definition}

\newtheorem*{remark}{Remark}

\numberwithin{equation}{section}

\begin{document}

\title{Spark-level sparsity and the $\ell_1$ tail minimization}

\author{Chun-Kit Lai}
\address{Department of Mathematics, San Francisco State University\\
San Francisco, CA 94132}
\email{cklai@sfsu.edu}

\author{Shidong Li}
\address{Department of Mathematics, San Francisco State University\\
San Francisco, CA 94132}
\email{shidong@sfsu.edu}

\author{Daniel Mondo}
\address{Department of Mathematics, San Francisco State University\\
San Francisco, CA 94132}
\email{dmondo@mail.sfsu.edu}

\begin{abstract}

Solving compressed sensing problems relies on the properties of sparse signals.  It is commonly assumed that the sparsity $s$ needs to be less than one half of the spark of the sensing matrix $A$, and then the unique sparsest solution exists, and recoverable by $\ell_1$-minimization or related procedures.  We discover, however, a measure theoretical uniqueness exists for nearly spark-level sparsity from compressed measurements $Ax=b$. Specifically, suppose $A$ is of full spark with $m$ rows, and suppose $\tfrac{m}{2} < s < m$. Then the solution to $Ax=b$ is unique for $x$ with $\|x\|_0\leq s$ up to a set of measure 0 in every $s$-sparse plane.  This phenomenon is observed and confirmed by an $\ell_1$-tail minimization procedure, which recovers sparse signals uniquely with $s > \tfrac{m}{2}$ in thousands and thousands of random tests.   We further show instead that the mere $\ell_1$-minimization would actually fail if $s > \tfrac{m}{2}$ even from the same measure theoretical point of view.


\end{abstract}
\maketitle

\section{Introduction} \label{S:introduction}

Compressed sensing is a problem that arises very naturally in signal processing applications. A sparse signal $x\in \mathbb{C}^N$ (a vector consisting of only $s<<N$ non-zero entries) is detected by a sensing matrix $A \in \C^{m \times N}$ with $m < < N$. The goal is to recover the exact sparse vector $x$ from the clearly under-determined/compressed measurements $b= A x$. In applications, the sensing matrix $A$ is treated as a physical device taking linear
measurements of our signal $x$, and we then think of $b$ as an undersampled measurement of $x$.  
Representative works include, for instance, \cite{candes2006robust} and \cite{donoho2006high}. One may also find, e.g., \cite{baraniuk2007compressive}, \cite{cande2008introduction}, \cite{candes2011compressed}, \cite{candes2006compressive}, \cite{candes2006robust}, \cite{candes2006near}, \cite{foucart2013mathematical}, \cite{fornasier2011compressive}, \cite{rauhut2010compressive}, and \cite{romberg2008imaging} for comprehensive descriptions. For some visualizable applications of compressed sensing theory, there are among others, \cite{candes2006robust}, \cite{haldar2011compressed}, \cite{lustig2007sparse} on  magnetic resonance imaging (MRI), and \cite{ender2010compressive}, \cite{fannjiang2010compressed}, \cite{herman2009high}, \cite{potter2010sparsity}, etc. on radar imaging.

\medskip

The task of recovering $x$ can be recast as a minimization problem:
\begin{equation}\label{P_0}
(\ell_0): \ \  \underset{x} {\text{min}} \ \|x\|_0 \ \ \text{subject to} \ b = Ax.
 \end{equation}
This is an NP-hard combinatorial problem (c.f. \cite{ge2011note}), and so the convex relaxation,
\begin{equation}\label{P_1}
(\ell_1): \ \  \underset{x} {\text{min}} \ \|x\|_1 \ \ \text{subject to} \ b = Ax,
 \end{equation}
is typically solved instead. Here $\|x\|_0= \#\{i:x_i\ne 0\}$ and $\|x\|_1 = \sum_{i=1}^N|x_i|$.  We will refer to \eqref{P_1} as the {\it basis pursuit} problem, \cite{chen2001atomic}.  Although numerical experiment suggests that the solution to \eqref{P_0} and \eqref{P_1} are equivalent, the first theoretical justification was given by Candes et al who proved that solutions to \eqref{P_0} and \eqref{P_1} are  equivalent with high probability for random matrices such as Gaussian random matrices  as long as $m \geq C \cdot s \ln N$ (\cite{candes2006compressive} \cite{candes2006near}), for some constant $C$. For a more general discussion about the relationship between $(\ell_0)$ and $(\ell_1)$, reader can refer to  \cite{donoho2003optimally}, \cite{donoho2006stable}, \cite{candes2006compressive}, \cite{rauhut2010compressive} or \cite{rudelson2008sparse}, \cite{rauhut2007random} and \cite{rauhut2008stability}, and the references therein. We will not be discussing methods other than basis pursuit for signal recovery in this paper, but they do exist: \cite{kim2007interior}, \cite{chen2001atomic}, \cite{dai2009subspace}, \cite{needell2009cosamp}, \cite{donoho2012sparse}, among others.

 \medskip

Naturally, the problem of compressed sensing makes sense only if a unique solution to the ($\ell_0$) problem exists and we wish to recover it by different algorithms. In general, the greater the number of non-zero entries of $x$ is, the more difficult the signal recovery it would be despite the uniqueness of the solution to ($\ell_0$). Furthermore,  it is well-known that we need to require $s\le m/2$, or $s$ less than half of the spark of $A$, for the unique $s$-sparse solution to exist  (see for example \cite{candes2006robust}, \cite{donoho2006high}).

\medskip

We discover however that the uniqueness of $(\ell_0)$ is possible with full Lebesgue measure even when $\tfrac{m}{2}<s<m$.  Surprisingly, we further demonstrate that a tail minimization procedures works well in this regime.  

\medskip

\subsection*{Compressed sensing with frames} Our investigation was originally motivated from a more involved problem of compressed sensing with sparse frame representations. Recall that a frame is a set of vectors $\left\{f_k\right\}$ in an inner product space $V$ such that there exists $0<A\le B<\infty$ with the property that
$$
\forall v \in V, \ A\|v\|^2 \leq \sum_{k \in \N} |\langle v, f_k \rangle |^2 \leq B \|v\|^2.
$$
For a general reference on frames, see, e.g.,  \cite{christensen2013introduction}. Compressed sensing with frames is explored in e.g.  \cite{aldroubi2011stability}, \cite{candes2011compressed}, \cite{chen2014null}, \cite{davenport2013signal}, \cite{giryes2013can}, \cite{gribonval2007highly}, and \cite{rauhut2008compressed}. In particular, \cite{rauhut2008compressed} establishes a condition under which signal recovery with frames succeeds. Let $D$ be the matrix whose column vectors are the frame vectors $\left\{f_k\right\}$. It is known that there are signals $f$ in practice that are naturally sparse in a frame representation \cite{zeng2016sparse}, namely $f=Dx$ and $x$ is sparse.  When coupled with compressed sensing methodologies, the under-determined matrix $A$ measures the signal $f$ by $b=ADx$.  The task is to recover $f$ from the known $b$, $A$ and $D$.

\medskip

The two typical approaches are the $\ell_1$-synthesis problem, see, e.g., \cite{elad2007analysis},
\begin{equation}\label{synthesis} \underset{x} {\text{min}} \ \|x\|_1 \ \ \text{subject to}\  b = ADx,
\end{equation}
or the $\ell_1$-analysis problem, e.g., \cite{candes2011compressed},
\begin{equation}\label{analysis} \underset{f} {\text{min}} \ \|D^*f\|_1 \ \ \text{subject to} \ b =Af.
\end{equation}
Note that when $D$ is actually a basis these methods are equivalent.

\medskip

An important object of study when solving these signal recovery problems is the error bound on a given recovered signal $\hat{f}$. The $\ell_1$-analysis approach has an error bound given by, under appropriate D-RIP condition of $A$, see, \cite{candes2011compressed}, \cite{liu2012compressed}:
\begin{equation}
\|\hat{f} - f\|_2 \leq C_0 \cdot \epsilon + C_1 \cdot \frac{\|{\tilde D}^*f - ({\tilde D}^*f)_s\|_1}{\sqrt{s}},
\end{equation}
where ${\tilde D}$ is a dual frame of $D$. Notice in particular that the error bound depends linearly on the \textit{tail} (smallest $d - s$ entries) of the signal: $\|{\tilde D}^*f - ({\tilde D}^*f)_s\|_1$.

\medskip

Since the error bound is directly proportional to the tail coefficients, minimizing the tail directly is a worthy topic of study, namely,
\[
\min_f \|{\tilde D}^*f - ({\tilde D}^*f)_s\|_1 \ \ \ \ \text{subject to}  \ b=Af.
\]
One immediately sees, however, the above minimization problem is non-convex.  Our next natural choice is to work with an iterative approach, where at each step, we identify an estimated support set $T$ and and solve the following \textit{tail minimization problem}:
\begin{equation} \label{tail-min}
\underset{f} {\text{min}} \ \|(D^*f)_{T^C}\|_1 \ \ \text{subject to} \ b =Af,
\end{equation}
where $T$ is the estimated support of $D^*f$, and $T^C$ is the complement of $T$.

\medskip

The procedure can have a number of variations.  The simplest case is similar to the iterative hard thresholding approach, \cite{blumensath2008iterative}. The hard thresholding step is to find an estimated support index $T$, and then solve the ``tail-min'' problem (\ref{tail-min}), followed by hard thresholding again for the next $T$, etc. Such a test was also performed over the traditional compressed sensing problem in a similar tail-minimization principle:
\begin{equation}\label{TMin}
\underset{x} {\text{min}} \ \|x_{T^C}\|_1 \ \ \text{subject to} \ b = Ax. \end{equation}
To be precise, we do the following:

\medskip

\noindent \textbf{Tail minimization algorithm}

\begin{enumerate}

\item Inputs: a matrix $A \in \C^{m \times N}$ and measurement vector $b \in \C^m$. Starting at iteration one, let the support $T_0 = \left\{1,...,N\right\}$. At the first  iteration, we solve the basis pursuit problem  \ref{P_0} to obtain an initial approximation $\hat{x}_1$ for our signal.

\item  Find the index set of the $s$-largest elements of $\hat{x}_1$, call it $T_1$. Solve the tail minimization problem \begin{equation}\label{TM1} \underset{x} {\text{min}} \ \|x_{T_1^C}\|_1 \ \ \text{subject to} \ b = Ax \end{equation}  or \begin{equation}\label{TM2} \underset{f} {\text{min}} \ \|(D^*f)_{T_1^C}\|_1 \ \ \text{subject to} \ b =Af.
\end{equation} 

\item Continue for k steps: find the index set of the $s$-largest elements $T_{k-1}$ of $\hat{x}_{k-1}$. Solve the tail minimization problem (\ref{TM1}) or (\ref{TM2}) with $T_{k-1}$. Call the solution at this iteration $\hat{x}_k$.

\item The algorithm terminates when successive iterations differ by a small enough constant:  $\| \hat{x}_k - \hat{x}_{k-1}\|_2 < \epsilon$.

\end{enumerate}

\medskip

%
%
%
%
%
%
%

\medskip

Figures \ref{TMPlot} and \ref{fig2} below demonstrate  extensive random tests and a statistics of successful recovery rate versus the sparsity level are presented for the ``tail-min'' procedure.  Also tested and plotted are the ($\ell_1$) basis pursuit results, for both conventional compressed sensing, and that with sparse frame representations.

\medskip

For the conventional compressed sensing problem, we let $A$ be a Gaussian random matrix. We apply both procedures, basis pursuit and tail minimization, to $s$-sparse signals $x$, and let $s$ increase until both procedures fail with certainty. Below are the results for $n = 1000$ trials.


\begin{figure}[h!]\label{TMPlot}
\centering
\includegraphics[width=.8\textwidth]{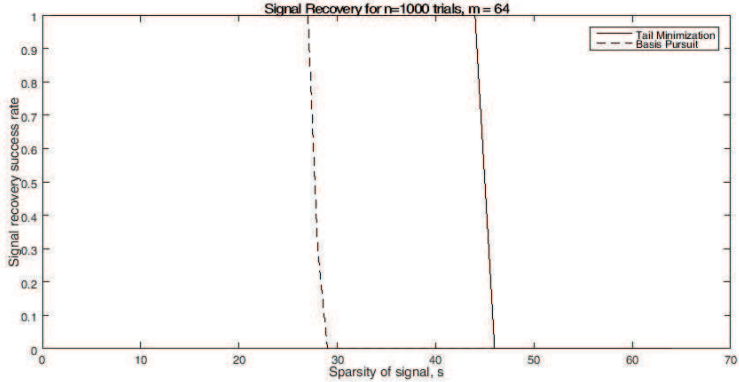}
\caption{Let $A \in \R^{64 \times 128}$ be a Gaussian random matrix, $x \in \R^{128}$. We plot the sparsity $s$ of the signal $x$ against the fraction of successful signal recovery for $n = 1000$ trials. The dotted line is the traditional basis pursuit (\ref{P_0}) and the solid line is the tail minimization procedure (\ref{TMin}). Signal recovery is considered a success if the relative error is less than a given tolerance, $\epsilon = 10^{-6}$.}
\end{figure}

\medskip

\begin{figure}[h!]
\centering
\includegraphics[width=.8\textwidth]{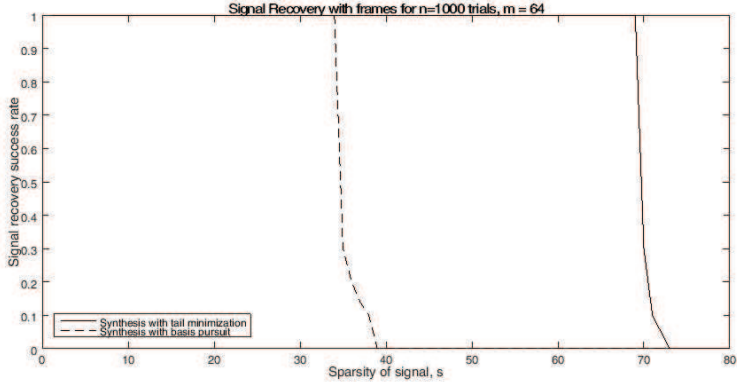}
\caption{Let $A \in \R^{64 \times 128}$ be a Gaussian random matrix, $D \in \C^{128 \times 256}$ a Fourier frame, and $x \in \R^{256}$. We plot the sparsity $s$ of the signal $x$ against the fraction of successful signal recovery for $n = 1000$ trials. The dotted line is the $\ell_1$-analysis problem (\ref{analysis}) and the solid line is the tail minimization procedure (\ref{tail-min}). Signal recovery is considered a success if the relative error is less than a given tolerance, $\epsilon = 10^{-6}$.}
\label{fig2}
\end{figure}

\medskip

It turns out that the procedure is not only doing well, but also greatly exceeding our expectation.  In fact, for vectors $x$ whose number of non-zeros $s$ greatly exceeding $m/2$ for a full rank matrix $A$ with $m$ rows, the ``tail-min'' procedure still recovers them all well, uniquely, in a massive amount of random testing.  It is widely known that the recovery of signals with $\frac{m}{2} < s < m$ is problematic since the ($\ell_0$) problem does not have unique solution. Specifically, consider a full-rank sensing matrix $A$ --- that is, rank$(A) = m$. The following is a well-known result (see e.g. \cite{bruckstein2009sparse} or \cite[Theorem 2.13]{FR13}):
\begin{lemma}
Every $s$-sparse vector is a unique solution of (1.2) if and only if every $2s$-columns of $A$ are linearly independent.
\end{lemma}
Consequently, if $s > \frac{m}{2}$, we cannot distinguish all $s$-sparse vectors in general. For example, if $A \in \C^{64 \times 128}$, then any $65$ columns must be linearly dependent. Hence, there exists $v\in \mbox{ker}(A) \setminus \left\{0\right\}$ such that $Av=0$ and the support of $v$ has at most 65 non-zero entries. Decompose $v = v_{S_1}+v_{S_2}$ where $|S_1| = 50$, $|S_2| = 15$, and $S_1 \cap S_2 = \emptyset$. Then $A(v_{S_1}) = A(-v_{S_2})$, which shows that there is not enough information to distinguish a 50-sparse vector $v_{S_1}$ and a 15-sparse vector $v_{S_2}$.


\medskip

\subsection*{Theoretical contribution} The statistically 100\% recovery with $\tfrac{m}{2} < s < m$ for the ``tail-min'' procedure in \textit{Figures  \ref{TMPlot}} and {\em \ref{fig2}} above initiates us to investigate the reason why it happens.  It turns out, a measure theoretical uniqueness solution exists for the ($\ell_0$) problem for $\tfrac{m}{2} < s < m$ with full spark matrices $A$ (see section 2 for a definition of full spark). We will prove that, given any $s$-sparse plane with $s<m$, the solution to $b=Ax$ is unique for $x$ with $\|x\|_0\leq s$ up to a set of measure 0 in the $s$-sparse coordinate plane (Theorem \ref{MainTheorem1}).

\medskip

On the other hand, we show that, when $A\in {\mathbb R}^{m\times N}$, the traditional basis pursuit fails on a set of infinite measures in some $s$-sparse plane (See Section 3).

\medskip

It comes to our attention during the investigation that \cite{wang2010sparse} presents a similar tail minimization procedure, though the analysis of the algorithm is quite different than what we go about. Our recoverability studies and the convergence analysis of the tail minimization algorithm will be presented in forthcoming articles.  
\medskip

For the rest of the paper, we will prove our measure theoretic uniqueness theorem in Section 2 and prove the failure of the basis pursuit in Section 3, both for the near spark-level sparsity.

\section{Uniqueness of Solution of $(\ell_0)$}

The \textit{spark} of a matrix $A$ is the smallest number of linearly dependent columns of $A$. The mathematical definition of spark can be written as
$$
\mbox{spark}(A) = \min\{\|x\|_0: Ax=0, \ x\ne 0\}.
$$
We say that $A \in \C^{m \times N}$ is \textit{full-spark} if any $m$ columns of $A$ are linearly independent, i.e. spark$(A) = m+1$. Full-spark frames exist almost everywhere under many probability models and it is a dense and open set in the Zariski topology in the sense that its complement is a finite union of zero sets of polynomials \cite{Full_Spark} . Given $A \in \C^{m \times N}$, for $T \subseteq \left\{1,...,m\right\}$ we define $A_T $ be the submatrix in ${\mathbb C}^{m\times|T|}$ formed by taking columns indexed by $T$. The coordinate plane indexed by $T$ is defined as
\begin{equation}
 H_T = \left\{(x_1,...,x_m): x_i \neq 0 \text{ for } i \in T;\  x_i = 0 \text{ for } i \notin T \right\}.
 \end{equation}
 We collect all possible $H_T$ of sparsity $s$ as follows:
\begin{equation}
\mathcal{H}_s = \left\{H_T: |T| = s\right\}.
 \end{equation}

\medskip

We know that a full spark matrix can recover all $s$-sparse signals by $(\ell_0)$ if $s<m/2$. The following theorem suggests that we can in principle recover almost all $s$ sparsity signals as long as $s<m$.
\begin{theorem}\label{MainTheorem1}
Let $\frac{m}{2} < s < m$ and $A \in \C^{m \times N}$ be full-spark. Let $H_{T_0}$ be any hyperplane in $\mathcal{H}_s$. Then for almost everywhere $x \in H_{T_0}$ (with respect to the $s$-dimensional Lebesgue measure on $H_{T_0}$), $x$ is the unique solution of $(\ell_0)$. That is, for $Ax = Av$, $v$ $s$-sparse, we have $x =v$ for almost everywhere $x \in H_{T_0}$.
\end{theorem}

\medskip

In order to prove the theorem, we require some lemmas:

\begin{lemma}\label{lem3.1}
Let $x$ be $s$-sparse. Then $x$ is the unique solution of $(\ell_0)$ if and only if
\begin{equation}\label{Eq3.1}
x \notin \underset{|T| \leq s}\bigcup H_T + (\ker{A} \setminus \left\{0\right\}).
\end{equation}
\end{lemma}

\begin{proof} Assume first that (\ref{Eq3.1}) holds. Suppose that $Ax = Av$ and $v$ is $s$-sparse. Then $A(x - v) = 0$. If $x \neq v$, then $x - v \in \ker{A} \setminus \left\{0\right\}$. But $v \in H_T$, so $x \in H_T +  (\ker{A} \setminus \left\{0\right\})$, which is a contradiction to (\ref{Eq3.1}).

\medskip

Conversely, suppose that $x$ belongs to the set in (\ref{Eq3.1}). Then $x = v+z$, $z \neq 0$, $v \in H_T$, $|T| \leq s$. We have that $A(x- v) = Az = 0$. This implies that $x$ cannot be the unique solution of $(\ell_0)$. This completes the proof. \end{proof}

\medskip

\begin{lemma}\label{lem3.2}
Let $A\in {\mathbb C}^{m\times N}$. Then $H_T \cap \ker{A} = \left\{0\right\}$ for any $H_T\in{\mathcal H}_s$ such that $s<\mbox{spark}(A)$.
\medskip
In particular, if $A$ is full spark, then all $H_T \cap \ker{A} = \left\{0\right\}$ for any $H_T\in{\mathcal H}_s$ and $s\le m$.
\end{lemma}

\begin{proof} Since $|T| \leq s < \mbox{spark}(A)$, columns in $A$ indexed by $T$ are linearly independent. Thus if  $Ax = 0$ and $x \in H_T$, we would have  $A_Tx = 0$, as any $s$ columns of $A$ are linearly independent. This implies that  $x = 0.$
 \end{proof}

\medskip

Because of  the previous lemma, when $A$ is full-spark, we have a direct sum between the subspace $H_T$ and $\ker{A}$, we shall denote it as $H_T\oplus\ker{A}$. Hence, for every $x \in H_{T}\oplus\ker{A}$, we have that
\begin{equation}\label{Eq3.2}
x = y_x+z_x
\end{equation}
 for unique $y_x \in H_T$ and unique $z_x \in \ker{A}$.

\medskip
The following lemma will be needed in the proof.
\begin{lemma}\label{lem3.3}
Given $H_{T_0}$ in the Theorem \ref{MainTheorem1} and any $T\subset{\mathcal H}_s$, suppose that $H_{T_0}\subset H_T\oplus \ker{A}$. Define a map
$$
P: H_{T_0} \longmapsto  \ker{A}, \  Px = z_x,
$$
Then $z_x\in H_{T_0\cup T}$, $P$ is linear and $P$ is one-to-one in $H_{T_0 \setminus T}$.
\end{lemma}

\begin{proof} As $x\in H_{T_0}$ and $y_x\in H_T$, $z_x = x-y_x\in H_{T_0\cup T}$. The fact that $P$ is linear follows by the uniqueness of the representation. Indeed, $$
P(x_1 + x_2) = P((y_{x_1}+z_{x_1})+(y_{x_2}+z_{x_2}) = P((y_{x_1}+y_{x_2})+(z_{x_1}+z_{x_2})) = P(x_1) + P(x_2).
$$
 To see that it is one-to-one on $H_{T_0\setminus T}$, we suppose that $ P(x_1) = P(x_2) $  for   $x_1, x_2 \in H_{T_0 \setminus T}$, we have $P(x_1 - x_2) = 0$. Let $u = x_1 - x_2 = y_u+z_u$. Then We have $u - y_u = P(x_1-x_2)=0$. Thus $u = y_u \in H_{T_0 \setminus T} \cap H_T$ since $u\in  H_{T_0 \setminus T}$ and $y_u\in H_T$. But $H_{T_0 \setminus T} \cap H_T = \left\{0\right\}$, so we have $u = 0$. Thus $x_1 = x_2$.
\end{proof}

\medskip

Now we move on to the proof of the theorem:

\begin{proof}[Proof of Theorem \ref{MainTheorem1}] Let $H_{T_0} \in \mathcal{H}_s$ be a plane. Using Lemma \ref{lem3.1}, $x \in H_{T_0}$ is \textit{not} the unique solution of $(\ell_0)$ if and only if $x$ belongs to the union defined in (\ref{Eq3.1}). We now decompose the the union in (\ref{Eq3.1}) into two sets:

$$ X_1: = H_{T_0} \cap \bigg[ \underset{|T| \leq s \atop \#(T_0 \cup T) \leq m}\bigcup H_T + (\ker{A} \setminus \left\{0\right\})\bigg], \ X_2: =   H_{T_0} \cap \bigg[ \underset{|T| \leq s \atop \#(T_0 \cup T) > m}\bigcup H_T + (\ker{A} \setminus \left\{0\right\})\bigg].
$$
We first claim that $X_1 = \emptyset$. Indeed, if $x\in H_{T_0}$ and $x = v +z$ for some $v+z\in H_T + \ker{A}\setminus\{0\}$, then $A(x - v) = A_{T_0\cup T}(x-v) = 0$, where the sparsity of $(x - v) = \#(T_0 \cup T) \leq m$. But $A$ is full-spark and so these vectors indexed by $T_0\cup T$ are linearly independent, so $x=v$. This forces $X_1$ is empty.

\medskip

Hence, the union in (\ref{Eq3.1}) is just $X_2$. Suppose that this union has positive Lebesgue measure in $H_{T_0}$. We enlarge the set by considering the zero vector in.
$$
 H_{T_0} \cap (H_T + \ker{A} \setminus \left\{0\right\}) \subseteq H_{T_0} \cap (H_T + \ker{A}).
 $$
 Consider the set
 \begin{equation} \label{F}
 \mathcal{F} = \underset{|T| \leq s \atop \#(T \cup T_0) >m}\bigcup H_{T_0} \cap (H_T + \ker{A}).
 \end{equation}
 Then $\mathcal{F}$ has positive measure in $H_{T_0}$. There exists some $T$ such that $|T| \leq s$ and $\#(T \cup T_0) > m$ with the property that $|H_{T_0} \cap (H_T + \ker{A}| > 0$ in $H_{T_0}$, where $|\cdot|$ denotes the Lebesgue measure.
Since $H_{T_0} \cap (H_T + \ker{A})$ is a subspace of $H_{T_0}$, positive measure implies that $H_{T_0} \cap (H_T + \ker{A}) = H_{T_0}$. Thus $H_{T_0} \subseteq H_T + \ker{A}$.

\medskip

By Lemma \ref{lem3.3}, $P: H_{T_0 \setminus T} \longmapsto \ker{A} \cap H_{T_0 \cup T}$ is one-to-one. We now compute the dimensions of the following subspaces. Since $H_{T_0 \setminus T}$ is a hyperplane, we have
\begin{equation}
\dim{(P(H_{T_0 \setminus T}))} = \dim{(H_{T_0 \setminus T})} = \#T_0 \setminus T.
 \end{equation}
 Note that $\ker{A} \cap H_{T_0 \cup T} = \left\{x: A_{T_0 \cup T}x_{T_0 \cup T} = 0 \right\}$. Because $A$ is full-spark, we have
\begin{equation}\label{eq3.3}
  \dim{(\ker{A} \cap H_{T_0 \cup T})} = \#(T_0 \cup T) - m.
\end{equation}
But $P(H_{T_0 \setminus T}) \subseteq \ker{A} \cap H_{T_0 \cup T}$,
\[
  \dim{(P(H_{T_0 \setminus T}))} \leq \dim{(\ker{A} \cap H_{T_0 \cup T})}.
\]
Then we have that
$$
\#(T_0 \setminus T) \leq \#(T_0 \cup T) - m.
$$
Note that  $\#(T_0 \setminus T) = \#(T_0 \cup T) - \#T.$ and $\#T\le s$,
  $$
  \#(T_0 \cup T) - s \leq \#(T_0 \cup T) - \#T \ = \ \#(T_0 \setminus T) \leq \#(T_0 \cup T) - m.
  $$
  This forces $s \geq m$. This is a contradiction to our assumption of $s$. Hence,  we must have $|\mathcal{F}| = 0$, completing the proof. \end{proof}

\medskip

Indeed, we can replace  $m$ by spark$(A)-1$. We have the same conclusion as in Theorem \ref{MainTheorem1}

\begin{theorem}\label{MainTheorem2}
Let $A \in \C^{m \times N}$ and let $s$ be such that $$\frac{\mbox{\rm spark}(A)-1}{2} < s < \mbox{\rm spark}(A)-1.$$
Let $H_{T_0}$ be any hyperplane in $\mathcal{H}_s$. Then for almost everywhere $x \in H_{T_0}$ (with respect to the $s$-dimensional Lebesgue measure on $H_T$), $x$ is the unique solution of $(\ell_0)$.
\end{theorem}

\begin{proof}
Following the same proof of Theorem \ref{MainTheorem1} until (\ref{eq3.3}), we note that
$$
\mbox{rank}(A_{T_0\cup T}) \ge \mbox{spark}(A)-1.
$$
 Otherwise, if $\mbox{rank}(A_{T_0\cup T})+1<\mbox{spark}(A)$, then $A_{T_0\cup T}$ will contain rank$(A_{T_0\cup T})$+1 linearly independent vectors, which contradicts to the definition of the rank. Hence,
$$
\dim{(\ker{A} \cap H_{T_0 \cup T})} = \#(T_0 \cup T) - \mbox{rank}(A_{T_0\cup T})\le \#(T_0 \cup T) -( \mbox{spark}(A)-1).
$$
Continuing the same proof and we finally arrives at $s\ge \mbox{spark}(A)-1$, contradicting our initial assumption.
\end{proof}

\begin{remark}
{\rm As we have indicated in the introduction, uniqueness solution is possible over the traditional regime of compressed sensing. From the numerical result through a random choice of initial $s$-sparse $x$, the measure theoretical conclusion shows that the probability of choosing $x$ as a non-unique solution of $(\ell_0)$ is 0. Therefore, recovery of $x$ in $m/2<s<m$ is possible.

\medskip

Comparing the basis pursuit which fails far before $m/2$, our $\ell_1$ tail minimization approach is not only capable of recovery vectors of near spark-level sparsity, but also going over to recover signals in $s>m/2$ regime.}
\end{remark}
%

\section{Failure of Basis Pursuit}

In this section, we provide a justification that the basis pursuit problem (1.2) has no unique solution on a significantly large set when $\frac{m}{2} < s < m$ and $A$ is a sensing matrix with real-valued entries. We first recall that the null space property is a necessary and sufficient condition for recovery of signals via basis pursuit.

\begin{definition}
A matrix $A \in \C^{m \times N}$ satisfies the null space proprty (NSP) with respect to $T$ if $\|v_T\|_1 < \|v_{T^C}\|_1$ for every $v \in \ker{A} \setminus \left\{0 \right\}$.
\end{definition}

\begin{theorem}
Every $s$-sparse vector is the unique solution to $(1.2)$ if and only if $A$ has the NSP for all $T$ with $|T| \leq s$.
\end{theorem}

 When the NSP of order $s$ holds, then every $s$-sparse vector is the unique solution to $(\ell_0)$. Because if $z$ is a solution to
\begin{equation*}
(\ell_0) \quad \quad \underset{x} {\text{min}} \ \|x\|_0 \ \ \text{subject to} \ b = Ax,
\end{equation*}
and $x$ is the solution to $(\ell_1)$, then $\|z\|_0 \leq \|x\|_0 \leq s$. But every $s$-sparse vector is the unique solution to $(\ell_1)$, $x=z$.

\begin{proposition}\label{prop_largeS_NSPfails}
If $\tfrac{m}{2} < s < m$, then there exists some $|T| \leq s$ such that  NSP  fails.
\end{proposition}

\begin{proof}
If the NSP holds on all $|T| \leq s$, then every $s$-sparse vector on $T$ is the unique solution of $(\ell_1)$ and hence $(\ell_0)$. This will require $m \geq 2s$, which is a contradiction.
\end{proof}

\begin{proposition}\label{prop_singleVecRecovery} \cite[Theorem 4.30]{FR13}
Given a matrix $A\in {\mathbb R}^{m\times N}$. A vector $x\in{\mathbb R}^N$ with support $S$ is the unique minimizer of $(\ell_1)$ if and only if $$
\left|\sum_{j\in S}\mbox{sgn}(x_j)v_j\right|<\|v_{S^{c}}\|_1, \ \forall v\in \mbox{ker}(A)\setminus\{0\},
$$
where
\[
\text{\rm sgn}(x) = \frac{x}{|x|} =  \begin{cases}
      1 & \textrm{ if $x > 0$} \\
      -1 & \textrm{ if $x < 0$} \\
        0 & \textrm{ if $x = 0$.} \\
   \end{cases}
\]
\end{proposition}

It turns out, as $\tfrac{m}{2} < s < m$, we can find some $T_0$ with $|T_0| \leq s$ such that the NSP with respect to $T_0$ fails.

\begin{theorem}\label{th3.1}
Let $A\in {\mathbb R}^{m\times N}$. If $\tfrac{m}{2} < s < m$, then there is a set of infinite measure on $H_{T_0}$ such that $(\ell_1)$ cannot succeed at recovery.
\end{theorem}

\begin{proof}
Since $\tfrac{m}{2} < s < m$, by  Proposition \ref{prop_largeS_NSPfails}, the NSP for some $T_0$ fails. There exists some $v \in \ker{A} \setminus \left\{0\right\}$ such that $\|v_{T_0}\|_1 \geq \|v_{T_0^C}\|_1$. By  Proposition \ref{prop_singleVecRecovery}, $x \in H_{T_0}$ is recovered as a unique solution of $(\ell_1)$ if and only if
\begin{equation}\label{eqn_singleVecRecovery}
 | \sum_{j \in T_0} \text{sgn}(x_j)v_j| < \|v_{T_0^C}\|_1.
\end{equation}
  Now, $\|v_{T_0}\|_1 \geq \|v_{T_0^C}\|_1$ for some $v \in \ker{A}\setminus\{0\}$. If we consider $x \in H_{T_0}$ such that sgn($x_j)v_j = |v_j|$, then the left hand side of (\ref{eqn_singleVecRecovery}) becomes $  \sum_{j \in T_0} |v_j| \geq \|v_{T_0^C}\|_1$. This shows all such $x$ are not recovered by $(\ell_1)$.

Let $(\epsilon_j)_{j \in T_0}$ be the sign such that $\epsilon_j v_j = |v_j|$. Then the failure set of $(\ell_1)$ contains $\left\{x \in H_{T_0}: \text{sgn}(x_j) = \epsilon_i \ \forall i \right\}$. This contains at least a quadrant in $H_{T_0}$, which is of infinite measure. Thus the statement holds.
\end{proof}

\medskip

\begin{remark}
The necessity of the {\em Proposition \ref{prop_singleVecRecovery}} is not true over the complex field \cite[Remark 4.29]{FR13}. It is not known for now whether Theorem \ref{th3.1} holds for ${\mathbb C}^{N}$ and complex matrices $A$. 

Nonetheless, the theorem shows that in the real case, $(\ell_0)$ and $(\ell_1)$ are not equivalent on a significant set of vectors when $s>m/2$.  We conclude the article by stating again that the $\ell_1$ tail minimization can still recover sparse signals for $m/2 < s < m$.
\end{remark}




\end{document}